\documentclass{article}

\usepackage[T2A]{fontenc}
\usepackage[utf8]{inputenc}
\usepackage[english]{babel}
\usepackage{geometry}
\usepackage{indentfirst}
\usepackage[linktoc=all]{hyperref}
\usepackage{tikz}
\usepackage{amsmath}
\usepackage{amssymb}
\usepackage{mathrsfs}
\usepackage{amsthm}
\usepackage{float}
\usepackage{tikz,listofitems}
\usepackage{enumitem}
\usepackage{scalerel,stackengine}
\usepackage{pict2e}
\usepackage{pdfpages}
\usepackage{natbib}
\usepackage{tabularx}

\theoremstyle{plain}
\newtheorem{theorem}{Theorem}
\newtheorem{lemma}{Lemma}

\newtheorem{problem}{Conjecture}
\newtheorem{example}{Example}

\theoremstyle{definition}
\newtheorem{definition}{Definition}

\begin{document}

\title{\LARGE \textbf{On the Construction of Recursively Differentiable Quasigroups and an Example of a Recursive $[4,2,3]_{26}$-Code}}
\author{
\large \textbf{Petr Klimov} \\
\normalsize \textit{Moscow Institute of Physics and Technology} \\ 
\normalsize  \texttt{email: peterklimov@yandex.ru}
}
\date{}

\maketitle

UDC 512.548.7+519.143+519.144

\begin{abstract}
In 1998, E. Couselo, S. González, V. T. Markov, and A. A. Nechaev introduced the notions of recursive codes and recursively differentiable quasigroups. They conjectured that recursive MDS codes of dimension $2$ and length $4$ exist over every finite alphabet of size $q \not\in \{2, 6\}$, and verified this conjecture in all cases except $q \in \{14, 18, 26, 42\}$. In 2008, V. T. Markov, A. A. Nechaev, S. S. Skazhenik, and E. O. Tveritinov resolved the case $q=42$ by providing an explicit construction. The present paper settles the outstanding case $q=26$. The construction rests upon methods for producing recursively differentiable quasigroups and recursive MDS codes via perfect cyclic Mendelsohn designs. Moreover, we sharpen several known bounds concerning the existence of recursively $n$-differentiable quasigroups of small orders.
\end{abstract}

\textbf{Keywords:} recursive codes, MDS codes, recursively differentiable quasigroups, perfect cyclic Mendelsohn designs.

\section{Introduction}

Let $\Omega = \{a_1, \ldots, a_q\}$ be a finite alphabet. A subset $K \subseteq \Omega^n$ is called a \textit{code of length $n$}, or simply an \textit{$n$-code over $\Omega$}.
It is customary to define the \textit{combinatorial dimension} of $K$ as $\log_{|\Omega|} |K|$.
Accordingly, a code of length $n$ and combinatorial dimension $k$ is referred to as an $[n,k]_\Omega$-code.
For any two words $\vec{u}, \vec{v} \in \Omega^n$, and in particular for any two codewords of $K$, the \textit{Hamming distance} $d(\vec{u}, \vec{v})$ is the number of coordinates in which they differ. The \textit{distance} of $K$, denoted $d(K)$, is the minimum Hamming distance between distinct codewords of $K$.
An $[n,k]_\Omega$-code with distance $d$ is thus called an $[n,k,d]_\Omega$-code, or equivalently, an $[n,k,d]_q$-code.
The classical \textit{Singleton bound} \cite{jmakv} asserts that for every code one has: $$d \le n - k + 1.$$
Codes attaining equality in the Singleton bound, that is, those for which $d = n - k + 1$, are known as \textit{maximum distance separable (MDS) codes}.

An important nontrivial case in the theory of MDS codes concerns the construction of $[4,2,3]_q$-codes \cite{Markov42}. Constructing such a code is equivalent to producing a pair of \textit{orthogonal Latin squares} on a set $\Omega$ of size $q$. Specifically, this means two $q \times q$ arrays in which each element of $\Omega$ appears exactly once in every row and every column, with the additional orthogonality condition that the ordered pairs of entries in corresponding cells are all distinct.
It is immediate that no such pair can exist for $q=2$. Likewise, by the resolution of Euler’s conjecture, proved in 1900 by Gaston Tarry \cite{tarry1900probleme}, no pair exists for $q=6$. In 1960, Bose, Shrikhande, and Parker established that for all other values of $q$ orthogonal Latin squares do indeed exist \cite{Bose_Shrikhande_Parker_1960}.

A \textit{left quasigroup} is a set $\Omega$ with a binary operation $*$ such that, for every $a,b \in \Omega$, the equation $$a * x = b$$ admits a unique solution $x \in \Omega$.
Analogously, a \textit{right quasigroup} is defined by the requirement that, for every $a,b \in \Omega$, the equation $$y * a = b$$ admits a unique solution $y \in \Omega$.
A \textit{quasigroup} is a groupoid that is simultaneously a left and a right quasigroup. Equivalently, the Cayley table of a finite groupoid on $\Omega$ forms a Latin square precisely when the groupoid is a quasigroup.
A quasigroup is said to be \textit{idempotent} if $a*a=a$ for all $a \in \Omega$.

In \cite{MarkovBase}, the notions of recursively $n$-differentiable quasigroups and full recursive codes were introduced. Given a quasigroup $(\Omega, *)$, one defines a recursive sequence of binary operations as follows:

$$a *_{-2} b = a,$$
$$a *_{-1} b = b,$$
$$\ldots$$
$$a *_n b = (a *_{n - 2} b) * (a *_{n - 1} b).$$

For $n \geq 0$, the groupoid $(\Omega, *_n)$ is called the \textit{$n$-th recursive derivative} of $(\Omega, *)$.
By definition, every quasigroup coincides with its $0$-th recursive derivative, while $(\Omega, *_1)$ is referred to simply as the \textit{recursive derivative} of $(\Omega, *)$.
A quasigroup $(\Omega, *)$ is said to be \textit{recursively $n$-differentiable} if $(\Omega, *_k)$ is a quasigroup for every $0 \le k \le n$.
In particular, every quasigroup is recursively $0$-differentiable, and a recursively $1$-differentiable quasigroup is usually called simply \textit{recursively differentiable}.
It is immediate from the definition that recursive $n$-differentiability implies recursive $k$-differentiability for $0 \le k < n$.

A code of length $n$ is said to be a \textit{full $k$-recursive code} if there exists a function $f: \Omega^k \to \Omega$ such that the code comprises all words $\vec{u} = (u_1, \ldots, u_n)$ satisfying $u_1, \ldots, u_k \in \Omega$ and, for $l > k$,
$$u_l = f(u_{l - k}, u_{l - k + 1}, \ldots u_{l-1}).$$
The combinatorial dimension of a full $k$-recursive code is precisely $k$.
The maximal length of full $k$-recursive MDS codes over an alphabet of size $q$ is denoted by $v^r(k,q)$; this quantity has been investigated extensively in \cite{MarkovBase, MarkovAdd}.
Of particular interest is the case $k = 2$, corresponding to the determination or estimation of $v^r(2,q)$.
In this case, the defining function $f$ reduces to a binary operation on $\Omega$.
Within this framework, the following theorem establishes the connection between recursive codes and recursively differentiable quasigroups.

\begin{theorem}\label{eq}\cite[Theorem 4]{MarkovBase}
     A full $2$-recursive code of length $n \ge 3$, specified by a function $f$, is MDS if and only if the corresponding groupoid $(\Omega, f)$ is a recursively $(n -3)$-differentiable quasigroup.
\end{theorem}

The most extensively examined nontrivial instance concerns determining the values of $q$ for which $v^r(2,q) \ge 4$. Equivalently, this amounts to establishing the existence of a recursive $[4,2,3]_q$-code. As noted above, such codes exist only for all $q \not\in \{2,6\}$. In \cite{MarkovBase}, the following conjecture was proposed.

\begin{problem}[Couselo-González-Markov-Nechaev]\label{hyp}
For every $q \not\in \{2, 6\}$ there exists a full recursive $[4, 2, 3]_q$-code. 
\end{problem}

By Theorem \ref{eq}, this conjecture is equivalent to the existence of a recursively differentiable quasigroup of order $q$.
Several constructions of such quasigroups are known, including those based on special transversals \cite{MarkovBase}, pseudogeometries \cite{Markov42}, linear recursive sequences \cite{Abashin}, and certain extension methods \cite{syrbu2023recursive}.

In \cite{MarkovBase}, E. Couselo, S. González, V.T. Markov, and A.A. Nechaev verified Conjecture \ref{hyp} for all $q$ with the exception of $q \in \{14, 18, 26, 42\}$. Subsequently, in \cite{Markov42}, Markov, Nechaev, Skazhenik, and Tveritinov established it for $q = 42$, providing the explicit construction.

In the present work, we establish Conjecture \ref{hyp} for $q = 26$ and provide an explicit construction of a recursively differentiable quasigroup of order $26$. The quasigroup obtained in our construction is, in fact, recursively $2$-differentiable. To this end, we develop a general methodology for constructing recursively $n$-differentiable quasigroups, and, consequently, recursive MDS codes of dimension $2$, via certain combinatorial designs. Furthermore, we refine the previously best-known bounds for $v^r(2, q)$ mentioned in \cite{MarkovBase, MarkovAdd}.

\section{Cyclic Mendelsohn Designs}

In the 1970s, Nathan Mendelsohn introduced a special class of cyclic block designs \cite{MENDELSOHN197763}. We briefly recall the basic definitions.

\begin{definition}
     A set $\{a_1, a_2,\ldots, a_k\}$ is said to be \textit{cyclically ordered} if it is equipped with the cyclic order $a_1 < a_2 < \ldots < a_k < a_1$.
\end{definition}

The cyclic order $a_1 < a_2 < \ldots < a_k < a_1$ on a cyclically ordered set $\{a_1, a_2, \ldots, a_k\}$ can be represented explicitly by writing the set as a sequence $(a_1, a_2, \ldots, a_k)$. In what follows we will write cyclically ordered sets in exactly this form.

\begin{definition}
   In a cyclically ordered set $(a_1, a_2,$ $ \ldots, a_k)$, the pair $a_i, a_{i + 1}$ is said to be \textit{consecutive}, with indices taken modulo $k$.
\end{definition}

\begin{definition}
    In a cyclically ordered set $(a_1, a_2,$ $ \ldots, a_k)$, the pair $a_i$, $a_{(i + t) }$ is said to be \textit{$t$-apart}, with indices taken modulo $k$.
\end{definition}

Therefore, a consecutive pair is precisely a $1$-apart pair.

\begin{definition}
    A \textit{cyclic Mendelsohn design} of type $(v, k, \lambda)$, or simply a $(v, k, \lambda)$-Mendelsohn design, is a pair $(X,B)$ where $X$ is a set of cardinality $v$ and $B$ is a collection of cyclically ordered subsets of $X$ with cardinality $k$ such that every ordered pair of distinct elements of $X$ appears consecutively in exactly $\lambda$ blocks.
\end{definition}

In a $(v, k, \lambda)$-Mendelsohn design $(X, B)$, the elements of $X$ are referred to as \textit{points}, while the elements of $B$ are referred to as \textit{blocks}.

\begin{definition}
    A cyclic Mendelsohn design $(X, B)$ of type $(v, k, \lambda)$ is said to be \textit{$l$-perfect} if, for every $t = 1, \dots, l$, each ordered pair of distinct points of $X$ is $t$-apart in exactly $\lambda$ blocks of $B$.
\end{definition}

Evidently, every $(v, k, \lambda)$-Mendelsohn design is $1$-perfect. 

\begin{definition}
    A cyclic Mendelsohn design $(X, B)$ of type $(v, k, \lambda)$ is said to be \textit{perfect} if it is $(k - 1)$-perfect.
\end{definition}

Cyclic Mendelsohn designs of type $(v, k, \lambda)$ are denoted by $(v, k,$ $ \lambda)$\textit{-MD}, whereas perfect Mendelsohn designs by $(v, k, \lambda)$-PMD.

\begin{example}
    Consider the set $X = \{0, 1, 2\}$ and the collection of blocks $B = \{(0, 1, 2), (0, 2, 1)\}$. One readily verifies that each ordered pair of distinct elements of $X$ is $t$-apart in exactly one of the blocks, for $t = 1, 2$. Hence $(X, B)$ is a $(3, 3, 1)$-PMD.
\end{example}

It is straightforward to verify that any $(v, k, \lambda)$-MD has exactly $\frac{\lambda v(v-1)}{k}$ blocks. Hence a necessary condition for the existence of a $(v, k, \lambda)$-PMD is $\lambda v(v -1) \bmod{k} \equiv 0$. While this condition is frequently sufficient, it does not always guarantee the existence \cite{BennettMD}. For instance, a $(6, 3, 1)$-PMD does not exist \cite{mendelsohn1971natural}.

Now let $k \ge 3$ and suppose $(X, B)$ is a $(v, k, 1)$-MD. Following \cite{lindner2003quasigroups}, we define the \textit{directed standard construction} of a groupoid on $X$ from $(X, B)$ by introducing a binary operation as follows:

\begin{enumerate}
    \item $a * a = a$ for every $a \in X,$
     \item $a * b = c$ for distinct $a, b \in X$ whenever there exists a block of $B$ of the form $(\ldots, a, b, c, \ldots).$
\end{enumerate}

The construction indeed yields a well-defined groupoid: for any pair of identical elements, the operation $a * a$ is uniquely determined. For distinct elements $a, b \in X$, the defining property of a Mendelsohn design ensures that there exists a unique block in which $a$ and $b$ appear consecutively, thereby uniquely specifying $a * b$.

\section{Construction of Recursively Differentiable Quasigroups via Combinatorial Designs}

We begin by establishing a periodicity property of the sequences under consideration.

\begin{theorem}
    Let $(X, *)$ be a finite right quasigroup. Then any sequence of the form $$a, b, a * b, b * (a * b), \ldots$$
    ---where each subsequent term, starting from the third, is the product of the preceding two---is periodic.

\end{theorem}
\begin{proof}
    Since $(X, *)$ is finite, there are only finitely many ordered pairs of elements; therefore, some consecutive pair must eventually repeat.
    Let $c, d$ denote the first such pair that appears twice in the sequence, assuming that two appearances are disjoint and separated by some distance; it is straightforward to verify that the argument extends naturally to the cases where the appearances are overlapping or adjacent. Suppose further that $(c, d) \neq (a, b)$. 
    
    Write the first two appearances of $c, d$ schematically as

    $$a, b, \dots, i, c, d, \dots, f, c,d, \dots$$

    If $i = f$, a contradiction arises, since then the pair $i, c$ would appear twice before $c, d$, violating the assumption that $c, d$ is the first repeated pair. Conversely, if $i \neq f$, we again obtain a contradiction, as in a right quasigroup the equation $y * c = d$ admits a unique solution $y$ for the given pair $c, d$.
    Hence, the first repeated pair must be the initial pair $a,b$. Writing out the first and second appearances of $a,b$ yields

    $$a, b, a * b, \ldots, a, b, a * b, \ldots$$

    and, since each subsequent element is uniquely determined by the previous two, the sequence is periodic with period equal to the distance between the two appearances of $a,b$.

\end{proof}

Therefore, any sequence of the form $a, b, a * b, b * (a * b), \ldots$ in a right quasigroup can be represented as a finite cycle $(a, b, \ldots)$. On such a cycle, one may impose a cyclic order $a < b < \ldots < a$, analogous to the order in a cyclically ordered set. It should be noted, however, that a cycle is not necessarily a cyclically ordered set, as elements may repeat. 

\begin{definition}
    A finite sequence $(a_1, a_2, \ldots, a_n)$ is said to be \textit{cyclically ordered} if it is endowed with the cyclic order $a_1 < a_2 < \ldots < a_n < a_1$.
\end{definition}

\begin{example}
    The cyclically ordered sequence $(2, 2, 1)$ is not a cyclically ordered set, since it contains repeated elements, which does not prevent us from considering the order among its elements.
\end{example}

Henceforth, we regard cycles of periodic sequences as cyclically ordered finite sequences, identifying those that differ only by a cyclic shift. For cycles, the notions of consecutiveness and $t$-apartness are defined in the same fashion as for cyclically ordered sets.

\begin{definition}
   A pair $a, b \in X$ is said to be \textit{consecutive in a cycle} $C$ of a periodic sequence if $C$ can be expressed as $C = (\ldots, a, b, \ldots)$ after some cyclic shift. We also consider the pair $a, a$ is consecutive in the cycle $(a)$ of length one.
   
\end{definition}

\begin{definition}
  The set of all cycles of sequences of the form 
  $$a, b, a * b, b * (a * b), \ldots$$
  in a right quasigroup $(X,*)$ is called the \textit{cyclic decomposition} of $(X,*)$ and is denoted by $Cycle(X, *)$
\end{definition}

\begin{lemma}\label{a}
    Let $(X, *)$ denote the groupoid obtained with the directed standard construction from a $(v, k, 1)$-MD with point set $X$ and block set $B$. Then $(X, *)$ is a right quasigroup. 
\end{lemma}
\begin{proof}
    For each $a \in X$, the equation $y *a = a$ has the unique solution $y = a$, since there is no block of form $(\ldots, a, a, \ldots
    )$, all elements in cyclically ordered block are distinct.

Consider a pair of distinct elements $a, b \in X$. There exists a unique block $C \in B$ in which they appear consecutively, say $C = (\ldots, c ,a ,b \ldots)$.  It follows that the equation $y * a = b$ admits the unique solution $y = c$.
\end{proof}

We now present a theorem establishing the connection between the cyclic decomposition of a right quasigroup and a cyclic Mendelsohn design.

\begin{theorem}\label{b}
    Let $(X, *)$ be a right quasigroup obtained with the directed standard construction from a $(v, k, 1)$-MD with point set $X$ and block set $B$. With regard to the blocks viewed as cyclically ordered sequences, we have $Cycle(X, *) / \{(a)\ |\ a \in X\} = B$, i.e., the set $B$ coincides with the set of cycles of $Cycle(X, *)$ once all cycles of length $1$ are omitted.
\end{theorem}
\begin{proof}
    Consider an arbitrary pair of distinct points $a, b \in X$. There exists a unique block $C \in B$ in which they appear consecutively, which can be written as
    $$C = (\ldots, a, b, c, \ldots). $$

    Since $C$ is cyclically ordered, we may relabel it as 
    $$C = (a, b, c, \ldots).$$

    By the definition of the directed standard construction, this yields $a * b = c $. To compute $b* c = b * (a * b)$, we again consult $C$ and take the element immediately following $b$ and $c$. Iterating this procedure produces 
    $$C = (a, b, a * b, b * (a * b), \ldots),$$
    with the procedure continuing until it cyclically returns to the initial pair $a, b$, after which it repeats. It follows that $C$ coincides precisely with the cycle of $Cycle(X,*)$ corresponding to the periodic sequence
    $$a, b, a * b,b * (a * b), \ldots$$

    Conversely, consider a cycle $C \in Cycle(X, *)$ of length greater than one that contains a pair of distinct elements $a, b$, that is $$C = (\ldots, a, b, \ldots).$$ Repeating the reasoning above, the cycle $C$ coincides with a unique block in $B$ when regarded as a cyclically ordered sequence.
\end{proof}

Because the directed standard construction uniquely determines a groupoid from a Mendelsohn design, Lemma \ref{a} and Theorem \ref{b} imply that the various $Cycle(X, *)$ provide natural extensions of the underlying $(v, k, 1)$-MD. Notably, whereas idempotence is not guaranteed for an arbitrary right quasigroup, the directed standard construction always yields an idempotent groupoid. Moreover, although blocks of a Mendelsohn design contain no repeated elements, cycles emerging from the cyclic decomposition of a right quasigroup may contain repetitions.

\begin{theorem}\label{c}
    Let $(X, *)$ be a right quasigroup obtained with the directed standard construction from a $(v, k, 1)$-MD with point set $X$ and block set $B$. Suppose that for some integer $t$ with $1 < t < k$, every pair of distinct points $a, b \in X$ appears $t$-apart in exactly one block of $B$. Then the $(t-2)$-th recursive derivative $(X,*_{t-2})$ is a left quasigroup, while the $(t-1)$-th recursive derivative $(X, *_{t - 1})$ is a right quasigroup.
\end{theorem}
\begin{proof}

By Theorem \ref{b}, for any pair of distinct points $a, b \in X$, there exists a unique block $C \in B$ of the form $C = (\ldots, a, b, \ldots)$, which can be written as $$C = (a, b, a * b, b * (a * b), \ldots).$$ Consequently, for each integer $d$ with $0 < d < k$, the block $C$ can be expressed in the form $$C = (a, b, \ldots, a *_d b, \ldots),$$ where precisely $d$ elements separate $a$ and $b$ from $a *_d b $ within the cycle $C$.

Let us consider an arbitrary pair of distinct points $a, b \in X$ and denote by $C$ the unique block in which they are $t$-apart, so that $$C = (\ldots, a, \ldots, b, \ldots).$$ By definition, exactly $t-1$ elements separate $a$ and $b$ within the cycle $C$.

We may represent $$C = (\ldots, y', a, x', \ldots, b, \ldots),$$ explicitly displaying the elements immediately preceding and following $a$. Observe that it may occur that $y' = b$, which does not affect the argument. Consequently, there are precisely $t-2$ elements between $x'$ and $b$ within the cycle $C$, and, as noted above, exactly $t-1$ elements between $a$ and $b$. Hence, we obtain:

\begin{itemize}
    \item $b = y' *_{t - 1} a$, and $y'$ is the unique solution to the equation $b = y *_{t - 1} a$.
    \item  $b = a *_{t - 2} x'$, and $x'$ is the unique solution to the equation $b = a *_{t - 2} x$.
\end{itemize}

The uniqueness of the solutions to the equations $b = y *_{t - 1} a$ and $b = a *_{t - 2} x$ for distinct $a, b \in X$ follows directly from the defining property of a Mendelsohn design, according to which the pairs $a, x$ and $y, a$ appear in a unique block $C$. Furthermore, the equations $a = y *_{t - 1} a$ and $b = a *_{t - 2} x$ admit the unique solution $x = y = a$, since there is no block of form $(\ldots, a, \ldots, a)$; in a cyclically ordered set all elements are distinct.

It thus follows that $(X, *_{t - 2})$ is a left quasigroup, whereas $(X, *_{t - 1})$ is a right quasigroup.
\end{proof}

The preceding theorem enables us to establish the main result, offering a systematic method for constructing recursively differentiable quasigroups from Mendelsohn combinatorial designs.
\begin{theorem}\label{d}
    Applying the directed standard construction to an $(n+2)$-perfect $(v, k, 1)$-Mendelsohn design, with $0 \le n \le k - 3$, produces a recursively $n$-differentiable quasigroup.
\end{theorem}
\begin{proof}

    Let $(X, B)$ be a $(v, k, 1)$-Mendelsohn design as in the theorem, and let $(X, *)$ denote the right quasigroup obtained from $(X, B)$ via the directed standard construction.    

    The right quasigroup $(X, *)$ is recursively $n$-differentiable if and only if all of its recursive derivatives $(X, *_k)$ are quasigroups for  $k = 0, 1, \ldots, n$. Equivalently, this holds precisely when each ($X, *_k$) is simultaneously a left and a right quasigroup for $k = 0, 1, \ldots, n$. 

    By Theorem \ref{c}, $(X, *_0)$ is a left quasigroup, since the $(v, k, 1)$-MD in the theorem's assumption is $(n+2)$-perfect, ensuring that for any two distinct points $a, b \in X$ there exists a unique block in $B$ in which they are $2$-apart. Consequently, $(X, *)$ is quasigroup.

    Similarly, under the assumptions of the theorem, for each $t = 2, \dots, n+2$ and any two distinct points $a, b \in X$, there exists a unique block in $B$ in which they are $t$-apart. It then follows from Theorem \ref{c} that $(X, *_{t-2})$ is a left quasigroup and $(X, *_{t-1})$ is a right quasigroup for all $t = 2, \dots, n+2$. Therefore, $(X, *)$ is recursively $n$-differentiable. 
\end{proof}

Let $(X, *)$ be the groupoid obtained with the directed standard construction from a $(v, k, 1)$-PMD. Then, by Theorem \ref{c} $(X, *_{k - 2})$ is a right quasigroup. Nevertheless, it cannot be a left quasigroup and therefore fails to be a quasigroup.

\begin{lemma}\label{aa}
    Let $(X, *)$ denote the right quasigroup obtained with the directed standard construction from a $(v, k, 1)$-MD with point set $X$ and block set $B$. Then, for every $a, b \in X$, we have $a *_{k - 2} b = a$ and consequently, $(X, *_{k - 2})$ is not a left quasigroup.
\end{lemma}
\begin{proof}
    This assertion follows from the observation that if a pair of distinct points $a, b \in X$ appears in a block $C$ in $B$ of length $k$, then by Theorem \ref{b} the block can be written as $$C =(a, b, a *_0 b, \ldots, a *_{k - 3} b),$$ which immediately implies that $a *
    _{k - 2} b = a$. Furthermore, under this operation, the equation $a *_{k - 2} x = b$ admits no solution when $a \ne b$.
\end{proof}

This lemma enables us to obtain an upper bound on the degree of recursive differentiability of quasigroups obtained from Mendelsohn designs.

\begin{lemma}\label{vanya}
    Suppose the directed standard construction yields a recursively $n$-differentiable quasigroup from a $(v, k, 1)$-MD with $k \ge 3$. Then $n \le k - 3$.
\end{lemma}
\begin{proof}
   This follows from the fact that, by Lemma \ref{aa}, the $(k-2)$-th recursive derivative is not a left quasigroup.

\end{proof}

Moreover, the bound from the lemma above is attained on perfect Mendelsohn designs.

\begin{theorem}\label{e}
    The directed standard construction applied to a $(v, k, 1)$-PMD with $k \ge 3$ yields a recursively $(k-3)$-differentiable quasigroup, whereas $(k-3)$ is its maximal degree of recursive differentiability.
\end{theorem}
\begin{proof}
It follows from Theorem \ref{d} and Lemma \ref{vanya}.
\end{proof}

\begin{theorem}\label{z}
    The existence of a $(v, k, 1)$-PMD guarantees that $v^r(2, q) \ge k$.
\end{theorem}
\begin{proof}
    It follows from Theorem \ref{e} and Theorem \ref{eq}.
\end{proof}

Furthermore, we are able to give an explicit characterization of the recursively $n$-differentiable quasigroups obtained via the directed standard construction from perfect cyclic Mendelsohn designs.

\begin{theorem}
    A recursively $n$-differentiable quasigroup $(X, *)$ arises from the directed standard construction applied to a $(v, n+3, 1)$-PMD if and only if it is idempotent and satisfies, for all distinct $a, b \in X$ and all $0 \le d < n+1$, the following conditions:
    
    \begin{enumerate}
        \item $a *_{n + 1} b = a$,
        \item $a *_{n + 2} b = b$,
        \item $a *_d b \ne a$.
    \end{enumerate}
\end{theorem}
\begin{proof}

A quasigroup obtained with the directed standard construction applied to a $(v, n+3, 1)$-PMD is, by definition and as ensured by Theorem \ref{e}, seen to possess all the properties required in the statement of the theorem.

    Conversely, assume that the conditions of the theorem hold. Consider the cyclic decomposition $Cycle(X, *)$, and focus on those cycles containing a consecutive pair of distinct elements $a, b$, namely of the form $$(a, b, a *_0 b, a *_1 b, \ldots).$$ Since $a *_d b \ne a$ for every $d < n + 1$, no such cycle can have length smaller than $n + 3$. On the contrary, the relations $a *_{n + 1} b = a$ and $a *_{n + 2} b = b$ force the length to be exactly $n + 3$. The condition $a *_d b \ne a$ for $d < n + 1$ further guarantees that no element is repeated within the cycle. Hence, once the trivial $1$-cycles are removed, the decomposition $Cycle(X, *)$ coincides with a set $B$ of blocks of a cyclic Mendelsohn design, and $(X, *)$ is obtained from $(X, B)$ with the directed standard construction. Moreover, the recursive $n$-differentiability of $(X, *)$ implies that for every $k = 0, 1, \ldots, n$ the equations $a *_k x = b$ and  $y *_k a = b$ admit unique solutions. This uniqueness ensures the existence of a single cycle in which $a$ and $b$ are separated by exactly $k + 1$ elements, as well as a unique cycle in which they are separated by precisely $k$ elements. It follows that the block set $B$ is precisely the collection of blocks of a perfect cyclic Mendelsohn design

\end{proof}

The theorem \ref{e}, along with its corollaries, allows for the immediate application of results from Mendelsohn design theory to the construction of recursively differentiable quasigroups and recursive MDS codes. While our primary focus lies in advancing Conjecture \ref{hyp}, it is equally important to refine the existing bounds on the maximal known degree of recursive differentiability.

We provide a series of key results in the theory of Mendelsohn combinatorial designs, as established in \cite{bennett1985conjugate, bennett1994incomplete, bennett1990perfect, mendelsohn1969combinatorial, zhang1990, abel1998direct, bennett1997existence, bennett1996packings, bennett1992constructions, bennett1992existence, bennett1998perfect, abel2000perfect, miao1995perfect, abel1998existence, abel2002existence}.

\begin{theorem}\label{L3}\cite[Theorem 5.2]{BennettMD} 
For all integers $v \ge 4$ with $v \equiv 0,1\bmod{4}$, there exists $(v, 4, 1)$-PMD, except for $v \in \{4,8\}$.
\end{theorem}

\begin{theorem}\label{L22}\cite[Theorem 5.3]{BennettMD} 
For all integers $v \ge 5$ with $v \equiv 0,1 \bmod{5}$, there exists $(v, 5, 1)$-PMD, except for $v \in \{6, 10, 15, 20\}$.
\end{theorem}

\begin{theorem}\label{L5}\cite[Theorems 5.4 and 5.5]{BennettMD}
For all integers $v \ge 6$ with $v \equiv 0, 1, 3, 4 \bmod{6}$ except for the following cases:
\begin{itemize}
    \item  $v \equiv 0 \bmod{6}$ and  $v \in \{6, 12,18,24,30,48,54,60,72,84,90,96,102,108,114, $ $132,138,150,$ $162,168,180,192,198\}$,
    \item$v \equiv 3 \bmod{6}$, and either $v \in \{207, 213, 219, 237, 243, 255, 297, 375, 411,$ $ 435, 453, 459,$ $ 471, 489, 495, 513, 519, 609, 615, 621, 657\}$, or $v \in [9, 135] \cup [153, 183]$,
    \item $v \equiv 4 \bmod{6}$, and either $v \in \{10, 16,22,34\}$, or $v \in [52, 148]$
\end{itemize}

$(v, 6 , 1)$-PMDs are known to exist.
\end{theorem}
\begin{theorem}\label{L6}\cite[Theorem 5.7]{BennettMD}
For all integers $v \geq 7$ with $v \equiv 0, 1 \bmod{7}$, except for $v \in \{14,15,21,22,28,35,36,$ $42,70,84,98,99,126,140,$ $141,147,148,154,182,$ $183,196,238,$ $245,273,294\}$, there exists a $(v, 7, 1)$-PMD.
\end{theorem}

Taking into account the previously established possibility of constructing recursively differentiable quasigroups from perfect Mendelsohn cyclic designs, these results lead us to the following theorem. 

\begin{theorem}\label{hu} The following lower bounds on the maximum length of recursive MDS-codes hold:
\begin{itemize}
    \item For all $q\equiv 0,1\bmod{4}$ with $q \ge 4$, except $q\in \{4,8\}$, one has $v^r(2,q)\geq 4$.
    \item For all $q\equiv 0,1\bmod{5}$ with $q \ge 5$, except $q\in \{6, 10, 15, 20\}$, one has $ v^r(2,q) \geq 5$.
    \item For all $q \equiv 0, 1, 3, 4 \bmod{6}$ with $q \ge 6$, except for

\begin{itemize}
    \item  $q \equiv 0 \bmod{6}$ and  $q \in \{6, 12,18,24,30,48,54,60,$ $72,$ $84,$ $90,$ $96,$ $102,$ $108,$ $114,$ $132,138,$ $150,162,168,180,192,198\}$,
    \item$q \equiv 3 \bmod{6}$ and either $q \in \{207,213,$ $219,$ $237,$ $243,$ $255,$ $297,375,411,435,$ $453,459,$ $471,489,495,513,519,609,615,621,657\}$, or $q \in [9, 135] \cup [153, 183]$,
    \item $q \equiv 4 \bmod{6}$ and either $q \in \{10, 16,22,34\}$, or $q \in [52, 148]$,
\end{itemize}
we have  $ v^r(2,q) \geq 6$.

    \item For all  $q\equiv 0, 1 \bmod{7}$, $q \geq 7$, except for 
$q \in \{14,15,21,22,28,$ $35,$ $36,$ $42,$ $70,84,98,99,126,$ $140,141,147,148,154,182,183,196,238,$ $245,$ $273,294\}$, we have $ v^r(2,q) \geq 7$.

\end{itemize}
\end{theorem}
\begin{proof}
It follows from Theorems \ref{z}, \ref{L3}, \ref{L22}, \ref{L5}, \ref{L6}.
\end{proof}

In \cite{MarkovAdd} and \cite{syrbu2022recursively}, tables were presented containing the best known lower bounds for $v^r(2,q)$, and the maximal known degrees of recursive differentiability of quasigroups of order $q$ for $q \le 100$ and $q \le 200$, respectively. It should be noted that these quantities are related by Theorem \ref{eq}: if a quasigroup of order $q$ with degree recursive differentiability $r(q)$ is known, then $v^r(2, q) \ge r(q) + 3$, and conversely. For this reason, the aforementioned tables coincide for the same values of $q$, up to an additive shift by $3$. In what follows, we provide a table in the form of the maximal known degrees of recursive differentiability of quasigroups of order $q$, obtained by applying Theorem \ref{hu} together with known results from \cite{MarkovBase, MarkovAdd, Markov42}.

In the table, the order $q$ is computed as the sum of the indices of the row and the column, with the pair $(0, 0)$ corresponding to $q = 100$. Each cell contains either the previously known bound (if it remains unchanged) or, in the case of an improvement achieved in this work, a pair consisting of the new bound followed by the earlier one.

\begin{table}[h!]
\centering
\begin{tabular}{|l|l|l|l|l|l|l|l|l|l|l|}
\hline
 & 0 & 1 & 2 & 3 & 4 & 5 & 6 & 7 & 8 & 9\\
\hline
0 (100) & 2 & $\infty$ & 0 & 1 & 2 & 3 & 0 & 5 & 6 & 7\\
\hline
10 & 1 & 9 & 1 & 11 & 0 & 1 & 14 & 15 & 0 & 17\\
\hline
20 & 2 & 2 & 1 & 21 & 2 & 23 & \textbf{2}/0 & 25 & \textbf{3}/2 & 17\\
\hline
30 & \textbf{2}/1 & 29 & 30 & 1 & 1 & 3 & \textbf{3}/1 & 35 & 1 & 2\\
\hline
40 & \textbf{3}/1 & 39 & \textbf{3}/1 & 41 & 2 & \textbf{2}/1 & \textbf{3}/1 & 45 & 1 & 47\\
\hline
50 & 4 & \textbf{2}/1 & 2 & 51 & 3 & 3 & 5 & 4 & 4 & 57\\
\hline
60 & 2 & 59 & 3 & 5 & 62 & 4 & 3 & 65 & 3 & 3\\
\hline
70 & 4 & 69 & 6 & 71 & 3 & 3 & 3 & 5 & 4 & 77\\
\hline
80 & 5 & 79 & 3 & 81 & 4 & 4 & 4 & 3 & 6 & 87\\
\hline
90 & 3 & 5 & 4 & 3 & 4 & 4 & 4 & 95 & 4 & 7\\
\hline
\end{tabular}
\caption{New bounds on the minimum known degree of recursive differentiability for quasigroups of order at most $100$.}
\end{table}

\section{Construction of a Recursively Differentiable Quasigroup of Order 26}

As demonstrated by Theorem \ref{hu} and by the table at the end of the preceding section, we have provided an affirmative answer to Conjecture \ref{hyp} for $q = 26$. Furthermore, we have established not merely the existence of a recursively differentiable quasigroup of order $26$, but in fact the existence of a recursively $2$-differentiable quasigroup of order $26$.

We now present its explicit construction, relying on the results of \cite{bennett1997existence}.

Let us consider the set $X=\{0, 1, \ldots, 20, \infty_1, \ldots, \infty_5\}$. Define a unary operation $\circ$ on $X$ by setting $\circ(i)=i+1 \bmod (21)$ for $i=0,..., 20$, $\circ(\infty_j)=\infty_j$ for $j=1,..., 5$.
Next, consider the collection of blocks $B$ obtained by applying $\circ$ elementwise to the initial set of blocks $$(0,1,14,20,19),(\infty_1,0,9,16,6),(\infty_2,0,10,8,12),$$$$(\infty_3,0,16,10,1),(\infty_4,0,17,4,9),(\infty_5,0,18,11,14),$$ as illustrated in Figure 1.

To this collection of blocks we adjoin $\{(\infty_1, \infty_2, \infty_3, \infty_4, \infty_5),$ $(\infty_1, \infty_3, \infty_5, \infty_2, \infty_4)$, $(\infty_1, \infty_4, \infty_2, \infty_5, \infty_3),$ $(\infty_1, \infty_5, \infty_4, \infty_3, \infty_2)\}$. It can be verified directly that $X$ as the set of points, together with the resulting collection of blocks, forms a $(26, 5, 1)$-PMD. Applying the directed standard construction to this design, we obtain by Theorem \ref{e} a recursively $2$-differentiable quasigroup of order $26$.
\\ \\ \\

\begin{figure}[!h]
    \centering
\setlength{\unitlength}{5cm}
\begin{picture}(2.1,1.9)(-1.05,-1.05)

 \linethickness{3pt}
 \put(0,1){\circle*{0.02}}
 \put(0.04,1){\!$\infty_1$}

 \put(-0.95,0.31){\circle*{0.02}}
\put(-1.08,0.31){$\infty_2$}

 \put(-0.59,-0,81){\circle*{0.02}}
 \put(-0.72,-0.81){$\infty_3$}

 \put(0.59,-0,81){\circle*{0.02}}
 \put(0.62,-0.81){$\infty_4$}

 \put(0.95,0.31){\circle*{0.02}}
 \put(0.98,0.31){$\infty_5$}

 \Line (0,1)(-0.95,0.31)
 \Line(-0.95,0.31)(-0.59,-0,81)
 \Line(-0.59,-0,81)(0.59,-0,81)
 \Line(0.59,-0,81)(0.95,0.31)
 \Line(0.95,0.31)(0,1)

 \put(0,-1){\arc[54,126]{1.62}}
 \put(0.95,-0.31){\arc[126,198]{1.62}}
 \put(0.59,0.81){\arc[198,270]{1.62}}
 \put(-0.59,0.81){\arc[270,342]{1.62}}
 \put(-0.95,-0.31){\arc[342,414]{1.62}}
 \linethickness{2pt}
 \polygon(0,1)(0,0.5)( -.22, -.45)( .50, .037)(-.49, -.11)
 \linethickness{1pt}
 \polygon(0,0.5)( -.15, .48)( .43, -.25)( .15, .48)(.28, .41)
 \put( 0, .5){\circle{0.02}\,{\scriptsize 0}}
 \put( -.15, .48){\circle{0.02}\,{\scriptsize 1}}
 \put( -.28, .41){\circle{0.02}\,{\scriptsize 2}}
 \put( -.39, .31){\circle{0.02}\,{\scriptsize 3}}
 \put( -.47, .18){\circle{0.02}\,{\scriptsize 4}}
 \put( -.50, .037){\circle{0.02}\,{\scriptsize 5}}
 \put( -.49, -.11){\circle{0.02}\,{\scriptsize 6}}
 \put( -.43, -.25){\circle{0.02}\,{\scriptsize 7}}
 \put( -.34, -.37){\circle{0.02}\,{\scriptsize 8}}
 \put( -.22, -.45){\circle{0.02}\,{\scriptsize 9}}
 \put( -.075, -.49){\circle{0.02}\,{\scriptsize 10}}
 \put( .075, -.49){\circle{0.02}\,{\scriptsize 11}}
 \put( .22, -.45){\circle{0.02}\,{\scriptsize 12}}
 \put( .34, -.37){\circle{0.02}\,{\scriptsize 13}}
 \put( .43, -.25){\circle{0.02}\,{\scriptsize 14}}
 \put( .49, -.11){\circle{0.02}\,{\scriptsize 15}}
 \put( .50, .037){\circle{0.02}\,{\scriptsize 16}}
 \put( .47, .18){\circle{0.02}\,{\scriptsize 17}}
 \put( .39, .31){\circle{0.02}\,{\scriptsize 18}}
 \put( .28, .41){\circle{0.02}\,{\scriptsize 19}}
 \put( .15, .48){\circle{0.02}\,{\scriptsize 20}}
\end{picture}\caption{Block construction scheme.}
\end{figure}
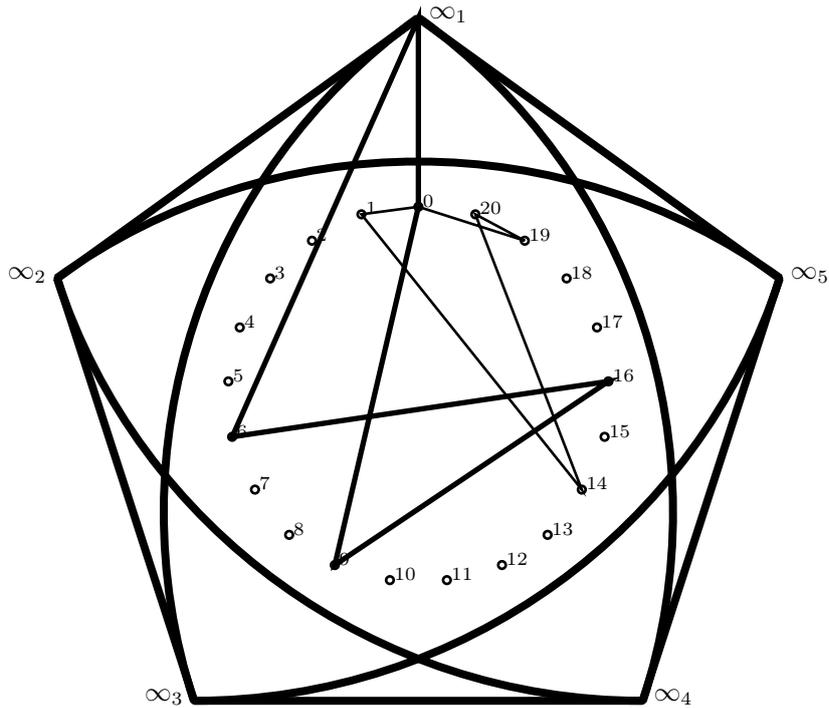

\section{A Recursively Differentiable Quasigroup of Order 26}

We now give an Cayley table of a recursively $2$-differentiable quasigroup of order $26$ obtained by the construction above, together with the Cayley tables of its recursive derivatives (which are themselves quasigroups). The alphabet is indexed by the digits $0$--$9$ and the letters $a$--$p$.
\newpage

\vspace*{1cm}
A recursively $2$-differentiable quasigroup of order $26$:
\vspace*{1cm}

\setlength{\tabcolsep}{2.5pt}

\begin{figure}[!h]
\centering
\begin{tabular}{cccccccccccccccccccccccccc}
0&e&3&p&m&o&5&i&d&g&8&l&n&j&h&6&a&4&b&2&1&f&9&k&c&7\\
2&1&f&4&p&m&o&6&j&e&h&9&l&n&k&i&7&b&5&c&3&g&a&0&d&8\\
4&3&2&g&5&p&m&o&7&k&f&i&a&l&n&0&j&8&c&6&d&h&b&1&e&9\\
e&5&4&3&h&6&p&m&o&8&0&g&j&b&l&n&1&k&9&d&7&i&c&2&f&a\\
8&f&6&5&4&i&7&p&m&o&9&1&h&k&c&l&n&2&0&a&e&j&d&3&g&b\\
f&9&g&7&6&5&j&8&p&m&o&a&2&i&0&d&l&n&3&1&b&k&e&4&h&c\\
c&g&a&h&8&7&6&k&9&p&m&o&b&3&j&1&e&l&n&4&2&0&f&5&i&d\\
3&d&h&b&i&9&8&7&0&a&p&m&o&c&4&k&2&f&l&n&5&1&g&6&j&e\\
6&4&e&i&c&j&a&9&8&1&b&p&m&o&d&5&0&3&g&l&n&2&h&7&k&f\\
n&7&5&f&j&d&k&b&a&9&2&c&p&m&o&e&6&1&4&h&l&3&i&8&0&g\\
l&n&8&6&g&k&e&0&c&b&a&3&d&p&m&o&f&7&2&5&i&4&j&9&1&h\\
j&l&n&9&7&h&0&f&1&d&c&b&4&e&p&m&o&g&8&3&6&5&k&a&2&i\\
7&k&l&n&a&8&i&1&g&2&e&d&c&5&f&p&m&o&h&9&4&6&0&b&3&j\\
5&8&0&l&n&b&9&j&2&h&3&f&e&d&6&g&p&m&o&i&a&7&1&c&4&k\\
b&6&9&1&l&n&c&a&k&3&i&4&g&f&e&7&h&p&m&o&j&8&2&d&5&0\\
k&c&7&a&2&l&n&d&b&0&4&j&5&h&g&f&8&i&p&m&o&9&3&e&6&1\\
o&0&d&8&b&3&l&n&e&c&1&5&k&6&i&h&g&9&j&p&m&a&4&f&7&2\\
m&o&1&e&9&c&4&l&n&f&d&2&6&0&7&j&i&h&a&k&p&b&5&g&8&3\\
p&m&o&2&f&a&d&5&l&n&g&e&3&7&1&8&k&j&i&b&0&c&6&h&9&4\\
1&p&m&o&3&g&b&e&6&l&n&h&f&4&8&2&9&0&k&j&c&d&7&i&a&5\\
d&2&p&m&o&4&h&c&f&7&l&n&i&g&5&9&3&a&1&0&k&e&8&j&b&6\\
9&a&b&c&d&e&f&g&h&i&j&k&0&1&2&3&4&5&6&7&8&l&n&p&m&o\\
a&b&c&d&e&f&g&h&i&j&k&0&1&2&3&4&5&6&7&8&9&p&m&o&l&n\\
g&h&i&j&k&0&1&2&3&4&5&6&7&8&9&a&b&c&d&e&f&o&l&n&p&m\\
h&i&j&k&0&1&2&3&4&5&6&7&8&9&a&b&c&d&e&f&g&n&p&m&o&l\\
i&j&k&0&1&2&3&4&5&6&7&8&9&a&b&c&d&e&f&g&h&m&o&l&n&p\\
\end{tabular}
\end{figure}
\newpage
\vspace*{1cm}

A recursive derivative of the aforementioned quasigroup:

\vspace*{1cm}

\begin{figure}[!h]

\centering
\begin{tabular}{cccccccccccccccccccccccccc}
0&k&g&a&d&h&7&l&o&6&c&5&b&i&p&n&1&9&e&m&2&3&j&f&8&4\\
3&1&0&h&b&e&i&8&l&o&7&d&6&c&j&p&n&2&a&f&m&4&k&g&9&5\\
m&4&2&1&i&c&f&j&9&l&o&8&e&7&d&k&p&n&3&b&g&5&0&h&a&6\\
h&m&5&3&2&j&d&g&k&a&l&o&9&f&8&e&0&p&n&4&c&6&1&i&b&7\\
d&i&m&6&4&3&k&e&h&0&b&l&o&a&g&9&f&1&p&n&5&7&2&j&c&8\\
6&e&j&m&7&5&4&0&f&i&1&c&l&o&b&h&a&g&2&p&n&8&3&k&d&9\\
n&7&f&k&m&8&6&5&1&g&j&2&d&l&o&c&i&b&h&3&p&9&4&0&e&a\\
p&n&8&g&0&m&9&7&6&2&h&k&3&e&l&o&d&j&c&i&4&a&5&1&f&b\\
5&p&n&9&h&1&m&a&8&7&3&i&0&4&f&l&o&e&k&d&j&b&6&2&g&c\\
k&6&p&n&a&i&2&m&b&9&8&4&j&1&5&g&l&o&f&0&e&c&7&3&h&d\\
f&0&7&p&n&b&j&3&m&c&a&9&5&k&2&6&h&l&o&g&1&d&8&4&i&e\\
2&g&1&8&p&n&c&k&4&m&d&b&a&6&0&3&7&i&l&o&h&e&9&5&j&f\\
i&3&h&2&9&p&n&d&0&5&m&e&c&b&7&1&4&8&j&l&o&f&a&6&k&g\\
o&j&4&i&3&a&p&n&e&1&6&m&f&d&c&8&2&5&9&k&l&g&b&7&0&h\\
l&o&k&5&j&4&b&p&n&f&2&7&m&g&e&d&9&3&6&a&0&h&c&8&1&i\\
1&l&o&0&6&k&5&c&p&n&g&3&8&m&h&f&e&a&4&7&b&i&d&9&2&j\\
c&2&l&o&1&7&0&6&d&p&n&h&4&9&m&i&g&f&b&5&8&j&e&a&3&k\\
9&d&3&l&o&2&8&1&7&e&p&n&i&5&a&m&j&h&g&c&6&k&f&b&4&0\\
7&a&e&4&l&o&3&9&2&8&f&p&n&j&6&b&m&k&i&h&d&0&g&c&5&1\\
e&8&b&f&5&l&o&4&a&3&9&g&p&n&k&7&c&m&0&j&i&1&h&d&6&2\\
j&f&9&c&g&6&l&o&5&b&4&a&h&p&n&0&8&d&m&1&k&2&i&e&7&3\\
g&h&i&j&k&0&1&2&3&4&5&6&7&8&9&a&b&c&d&e&f&l&o&m&p&n\\
8&9&a&b&c&d&e&f&g&h&i&j&k&0&1&2&3&4&5&6&7&o&m&p&n&l\\
a&b&c&d&e&f&g&h&i&j&k&0&1&2&3&4&5&6&7&8&9&m&p&n&l&o\\
4&5&6&7&8&9&a&b&c&d&e&f&g&h&i&j&k&0&1&2&3&p&n&l&o&m\\
b&c&d&e&f&g&h&i&j&k&0&1&2&3&4&5&6&7&8&9&a&n&l&o&m&p\\
\end{tabular}
\end{figure}
\newpage

\vspace*{1cm}

A second-order recursive derivative of the aforementioned quasigroup:
\vspace*{1cm}

\begin{figure}[!h]

\centering
\begin{tabular}{cccccccccccccccccccccccccc}
0&j&1&7&2&d&8&c&4&l&m&e&6&k&3&5&n&o&p&b&f&a&h&9&g&i\\
g&1&k&2&8&3&e&9&d&5&l&m&f&7&0&4&6&n&o&p&c&b&i&a&h&j\\
d&h&2&0&3&9&4&f&a&e&6&l&m&g&8&1&5&7&n&o&p&c&j&b&i&k\\
p&e&i&3&1&4&a&5&g&b&f&7&l&m&h&9&2&6&8&n&o&d&k&c&j&0\\
o&p&f&j&4&2&5&b&6&h&c&g&8&l&m&i&a&3&7&9&n&e&0&d&k&1\\
n&o&p&g&k&5&3&6&c&7&i&d&h&9&l&m&j&b&4&8&a&f&1&e&0&2\\
b&n&o&p&h&0&6&4&7&d&8&j&e&i&a&l&m&k&c&5&9&g&2&f&1&3\\
a&c&n&o&p&i&1&7&5&8&e&9&k&f&j&b&l&m&0&d&6&h&3&g&2&4\\
7&b&d&n&o&p&j&2&8&6&9&f&a&0&g&k&c&l&m&1&e&i&4&h&3&5\\
f&8&c&e&n&o&p&k&3&9&7&a&g&b&1&h&0&d&l&m&2&j&5&i&4&6\\
3&g&9&d&f&n&o&p&0&4&a&8&b&h&c&2&i&1&e&l&m&k&6&j&5&7\\
m&4&h&a&e&g&n&o&p&1&5&b&9&c&i&d&3&j&2&f&l&0&7&k&6&8\\
l&m&5&i&b&f&h&n&o&p&2&6&c&a&d&j&e&4&k&3&g&1&8&0&7&9\\
h&l&m&6&j&c&g&i&n&o&p&3&7&d&b&e&k&f&5&0&4&2&9&1&8&a\\
5&i&l&m&7&k&d&h&j&n&o&p&4&8&e&c&f&0&g&6&1&3&a&2&9&b\\
2&6&j&l&m&8&0&e&i&k&n&o&p&5&9&f&d&g&1&h&7&4&b&3&a&c\\
8&3&7&k&l&m&9&1&f&j&0&n&o&p&6&a&g&e&h&2&i&5&c&4&b&d\\
j&9&4&8&0&l&m&a&2&g&k&1&n&o&p&7&b&h&f&i&3&6&d&5&c&e\\
4&k&a&5&9&1&l&m&b&3&h&0&2&n&o&p&8&c&i&g&j&7&e&6&d&f\\
k&5&0&b&6&a&2&l&m&c&4&i&1&3&n&o&p&9&d&j&h&8&f&7&e&g\\
i&0&6&1&c&7&b&3&l&m&d&5&j&2&4&n&o&p&a&e&k&9&g&8&f&h\\
6&7&8&9&a&b&c&d&e&f&g&h&i&j&k&0&1&2&3&4&5&l&p&o&n&m\\
c&d&e&f&g&h&i&j&k&0&1&2&3&4&5&6&7&8&9&a&b&n&m&l&p&o\\
1&2&3&4&5&6&7&8&9&a&b&c&d&e&f&g&h&i&j&k&0&p&o&n&m&l\\
9&a&b&c&d&e&f&g&h&i&j&k&0&1&2&3&4&5&6&7&8&m&l&p&o&n\\
e&f&g&h&i&j&k&0&1&2&3&4&5&6&7&8&9&a&b&c&d&o&n&m&l&p\\
\end{tabular}
\end{figure}

\newpage

\section{Conclusion}
In the present work, the Couselo-González-Markov-Nechaev hypothesis is proved for the case $q = 26$, i.e., a full recursive $[4,2,3]_q$-code, or equivalently a recursively differentiable quasigroup of order $q$, is constructed. To this end, methods for constructing recursive MDS codes and recursively differentiable quasigroups by means of perfect cyclic Mendelsohn designs were introduced. 

The two remaining cases of the hypothesis, $q = 14$ and $q = 18$, can potentially be solved 
by adapting the methods used in the theory of perfect cyclic Mendelsohn designs, 
because recursively differentiable quasigroups admit a wider class of cyclic decompositions 
than just Mendelsohn designs, and applying similar methods to these decompositions may yield 
an even broader class of constructions than those obtained in this paper. 
At the same time, these cases may also be resolved by other classes of methods, 
in particular those used previously.

\textbf{Acknowledgments:} The author is sincerely grateful to Viktor Timofeyevich Markov for introducing him to the problem at the time, as well as his attentive engagement with the results that inspired the present paper. The main results of this paper were obtained in 2015 and remained unpublished for a long time. We also note that Figure 1 was drawn by Victor Timofeyevich Markov as a visualization of a construction sent to him by the author. He also expresses his profound gratitude to Aleksei Yakovlevich Kanel-Belov for helpful discussions on the matter.
\bibliography{refs.bib}

@article{MarkovBase,
url = {https://doi.org/10.1515/dma.1998.8.3.217},
title = {{Recursive MDS-codes and recursive differentiable quasigroups}},
author = {{E. Couselo, S. Gonzales, V. Markov, A. Nechaev}},
pages = {217--246},
volume = {8},
number = {3},
journal = {Discrete Math. Appl.},
doi = {10.1515/dma.1998.8.3.217},
year = {1998},
lastchecked = {2023-03-28}
}

@book{jmakv,
  title={Theory of error correction codes},
  author={{J. Mc-Williams, N. Sloan}},
  year={1979},
  publisher={Per. from English M .: Communication}
}

@book{tarry1900probleme,
  title={Le probl{\`e}me des 36 officiers},
  author={G. Tarry},
  year={1900},
  publisher={Secr{\'e}tariat de l'Association fran{\c{c}}aise pour l'avancement des sciences}
}

@article{Bose_Shrikhande_Parker_1960,
title={{Further Results on the Construction of Mutually Orthogonal Latin Squares and the Falsity of Euler’s Conjecture}},
volume={12},
DOI={10.4153/CJM-1960-016-5}, journal={Canadian Journal of Mathematics},
author={{R. Bose, S. Shrikhande, E. Parker}},
year={1960},
pages={189–203}}

@article{Markov42,
url = {https://doi.org/10.1007/s10958-009-9694-6},
title = {Pseudogeometries with clusters and an example of a recursive $[4, 2, 3]_{42}$-code},
author = {{V. Markov, A. Nechaev, S. Skazhenik,  E. Tveritinov}},
pages = {563--571},
volume = {163},
journal = {J Math Sci},
doi = {10.1007/s10958-009-9694-6},
year = {2009},
lastchecked = {2023-03-28}
}

@article{MarkovAdd,
url = {https://doi.org/10.1515/dma.2000.10.5.433},
title = {{Parameters of recursive MDS-codes}},
author = {{E. Couselo, S. Gonzales, V. Markov, A. Nechaev}},
pages = {433--454},
volume = {10},
number = {5},
journal = {Discrete Math. Appl.},
doi = {10.1515/dma.2000.10.5.433},
year = {2000},
lastchecked = {2023-03-28}
}

@article{Abashin,
url = {https://doi.org/10.1515/dma.2000.10.3.319},
title = {{Linear recursive MDS codes of dimensions 2 and 3}},
author = {A. Abashin},
pages = {319--332},
volume = {10},
number = {3},
journal = {Discrete Mathematics and Applications},
doi = {10.1515/dma.2000.10.3.319},
year = {2000},
lastchecked = {2025-07-25}
}

@article{syrbu2022recursively,
  title={On recursively differentiable k-quasigroups},
  author={{P. Syrbu, E. Kuznetsova}},
  journal={Buletinul Academiei de {\c{S}}tiin{\c{t}}e a Republicii Moldova. Matematica},
  volume={99},
  number={2},
  pages={68--75},
  year={2022}
}

@article{syrbu2023recursive,
  title={On recursive 1-differentiability of the quasigroup prolongations},
  author={{P. Syrbu, E. Kuznetsova}},
  journal={Buletinul Academiei de {\c{S}}tiin{\c{t}}e a Republicii Moldova. Matematica},
  volume={102},
  number={2},
  pages={102--109},
  year={2023}
}

@article{MENDELSOHN197763,
title = {Perfect cyclic designs},
journal = {Discrete Mathematics},
volume = {20},
pages = {63-68},
year = {1977},
issn = {0012-365X},
doi = {https://doi.org/10.1016/0012-365X(77)90043-7},
url = {https://www.sciencedirect.com/science/article/pii/0012365X77900437},
author = {N. Mendelsohn}
}

@article{BennettMD,
title = {{Recent progress on the existence of perfect Mendelsohn designs}},
journal = {Journal of Statistical Planning and Inference},
volume = {94},
number = {2},
pages = {121-138},
year = {2001},
note = {Second Shanghai Conf. on Designs, Codes and Finite Geometries},
issn = {0378-3758},
doi = {https://doi.org/10.1016/S0378-3758(00)00245-7},
url = {https://www.sciencedirect.com/science/article/pii/S0378375800002457},
author = {F. Bennett},
}

@article{mendelsohn1971natural,
  title={{A natural generalization of Steiner triple systems}},
  author={Mendelsohn, N.},
  journal={Computers in number theory},
  pages={323--338},
  year={1971},
  publisher={Academic Press New York}
}

@article{lindner2003quasigroups,
  title={Quasigroups constructed from cycle systems},
  author={C. Lindner},
  journal={Quasigroups and related systems},
  volume={10},
  number={1},
  pages={29--64},
  year={2003}
}

@article{bennett1985conjugate,
  title={{Conjugate orthogonal Latin squares and Mendelsohn designs}},
  author={F. Bennett},
  journal={Ars Combin},
  volume={19},
  pages={51--62},
  year={1985}
}

@article{bennett1994incomplete,
  title={{Incomplete perfect Mendelsohn designs with block size 4 and holes of size 2 and 3}},
  author={{F. Bennett, H. Shen, J. Yin}},
  journal={Journal of Combinatorial Designs},
  volume={2},
  number={3},
  pages={171--183},
  year={1994},
  publisher={Wiley Online Library}
}

@article{bennett1990perfect,
  title={{Perfect Mendelsohn designs with block size 4}},
  author={{F. Bennett, X. Zhang, L. Zhu}},
  journal={Ars Combinatoria},
  volume={29},
  pages={65--72},
  year={1990},
  publisher={CHARLES BABBAGE RES CTR PO BOX 272 ST NORBERT POSTAL STATION, WINNIPEG MB~…}
}

@article{mendelsohn1969combinatorial,
  title={Combinatorial designs as models of universal algebras},
  author={N. Mendelsohn},
  journal={Recent Progress in Combinatorics, Academic Press, New York},
  volume={123132},
  year={1969}
}

@article{zhang1990,
  title={{On the existence of $(v, 4, 1)$-PMD}},
  author={X. Zhang},
  journal={Ars Combinatoria},
  volume={29},
  pages={3--12},
  year={1990}
}

@article{abel1998direct,
  title={{Direct constructions for certain types of HMOLS}},
  author={{R. Abel, H. Zhang}},
  journal={Discrete mathematics},
  volume={181},
  number={1-3},
  pages={1--17},
  year={1998},
  publisher={Elsevier}
}

@article{bennett1997existence,
  title={{Existence of HPMDs with block size five}},
  author={{F. Bennett, Y. Chang, J. Yin, H. Zhang}},
  journal={Journal of Combinatorial Designs},
  volume={5},
  number={4},
  pages={257--273},
  year={1997},
  publisher={Wiley Online Library}
}

@article{bennett1996packings,
  title={{Packings and coverings of the complete directed multigraph with 3-and 4-circuits}},
  author={{F. Bennett, J. Yin}},
  journal={Discrete Mathematics},
  volume={162},
  number={1-3},
  pages={23--29},
  year={1996},
  publisher={Elsevier}
}

@article{bennett1992constructions,
  title={{Constructions of perfect Mendelsohn designs}},
  author={{F. Bennett, K. Phelps, C. Rodger, L. Zhu}},
  journal={Discrete mathematics},
  volume={103},
  number={2},
  pages={139--151},
  year={1992},
  publisher={Elsevier}
}

@article{bennett1992existence,
  title={{Existence of perfect Mendelsohn designs with $k= 5$ and $\lambda> 1$}},
  author={{F. Bennett, K. Phelps, C. Rodger, J. Yin, L. Zhu}},
  journal={Discrete mathematics},
  volume={103},
  number={2},
  pages={129--137},
  year={1992},
  publisher={Elsevier}
}

@article{bennett1998perfect,
  title={{Perfect Mendelsohn packing designs with block size five}},
  author={{F. Bennett, J. Yin, H. Zhang, R. Abel}},
  journal={Designs, Codes and Cryptography},
  volume={14},
  number={1},
  pages={5--22},
  year={1998},
  publisher={Springer}
}

@article{abel2000perfect,
  title={{Perfect Mendelsohn designs with block size six}},
  author={{R. Abel, F. Bennett, H. Zhang}},
  journal={Journal of statistical planning and inference},
  volume={86},
  number={2},
  pages={287--319},
  year={2000},
  publisher={Elsevier}
}

@article{miao1995perfect,
  title={{Perfect Mendelsohn designs with block size six}},
  author={{Y. Miao, L. Zhu}},
  journal={Discrete mathematics},
  volume={143},
  number={1-3},
  pages={189--207},
  year={1995},
  publisher={Elsevier}
}

@article{abel1998existence,
  title={{The existence of perfect Mendelsohn designs with block size 7}},
  author={{R. Abel, F. Bennett}},
  journal={Discrete mathematics},
  volume={190},
  number={1-3},
  pages={1--14},
  year={1998},
  publisher={Elsevier}
}

@article{abel2002existence,
  title={{Existence of Steiner seven-cycle systems}},
  author={{R. Abel, F. Bennett, G. Ge, L.Zhu}},
  journal={Discrete mathematics},
  volume={252},
  number={1-3},
  pages={1--16},
  year={2002},
  publisher={Elsevier}
}

\bibliographystyle{unsrt}

\end{document}